\documentclass{article}

\pdfoutput=1 

\usepackage[utf8]{inputenc}
\usepackage{amsmath}
\usepackage{amssymb}
\usepackage{amsthm}
\usepackage{natbib}

\usepackage[boxed]{algorithm2e}

\usepackage{graphicx}
\usepackage{pgfplots}
\usepackage{tikz}
\usetikzlibrary{shapes, patterns, decorations, fit, intersections}

\usepackage{xcolor}
\colorlet{drkblue}{blue!61.8!black}
\usepackage[bookmarks=false]{hyperref}
\hypersetup{
	pdfborder = {0 0 0},
	colorlinks=true,
	citecolor=drkblue,
	linkcolor=drkblue,
	urlcolor=drkblue,
	pdfstartview={XYZ null null 0.80}
}

\usetikzlibrary{arrows,automata}

\usepackage{paralist}
\usepackage{multirow}

\usepackage {array}
\newcolumntype{P}[1]{>{\centering\arraybackslash}p{#1}}  
\newcolumntype{M}[1]{>{\centering\arraybackslash}m{#1}}

\newcommand{\df}[1]{\textit{#1}}
\newcommand{\idf}[2]{\textit{#1}\index{#2}}

\let\R\relax
\newcommand*{\R}{\mathbb{R}} 
\newcommand*{\Z}{\mathbb{Z}}

\newcommand*{\dotprod}[2]{\langle#1,#2\rangle}

\newcommand*{\card}[1]{\lvert#1\rvert}

\newtheorem{theorem}{Theorem}
\newtheorem{lemma}[theorem]{Lemma}
\newtheorem{corollary}[theorem]{Corollary}
\newtheorem{proposition}[theorem]{Proposition}

\theoremstyle{definition}
\newtheorem{definition}[theorem]{Definition}
\newtheorem{remark}[theorem]{Remark}

\newtheorem{exercise}{}
\newtheorem{example}[theorem]{Example}

\newcommand*{\set}[1]{\{#1\}}
\newcommand*{\mr}{\mathrel}
\newcommand*{\mc}{\mathcal}

\newcommand*{\nin}{\not\in}
\newcommand*{\toto}{\rightrightarrows}

\renewcommand{\subset}{\subseteq}

\DeclareMathOperator*{\Top}{top}

\newcommand{\sminus}{\setminus}

\newcommand{\mbfC}{\mathbf{C}} 
\newcommand{\mbfD}{\mathbf{D}} 
\newcommand{\mcX}{\mathcal{X}}
\newcommand{\mbbC}{\mathbb{C}}
\newcommand{\X}{E}
\newcommand{\bY}{Y}
\newcommand{\bZ}{Z}
\newcommand{\bS}{S}
\newcommand{\bT}{T}
\newcommand{\x}{x}
\newcommand{\y}{y}

\newcommand{\rsu}{\mathrel{\rightarrowtail}}

\newcommand{\wpref}{{\mathrel{R}}}
\newcommand{\AM}{\textsf{AM}}

\newcommand{\ccf}{combinatorial choice function }

\makeatletter
\def\@xfootnote[#1]{%
	\protected@xdef\@thefnmark{#1}%
	\@footnotemark\@footnotetext}
\makeatother

	\title{Choice and Market Design\thanks{Battal Do\u{g}an gratefully acknowledges financial support from the British Academy/Leverhulme Trust (SRG1819\textbackslash 190133).}
	}
	\author{
		Samson Alva\footnote{University of Texas at San Antonio. Email: samson.alva@utsa.edu.}
		\quad and \quad
		Battal Do\u{g}an\footnote{University of Bristol. Email: battal.dogan@bristol.ac.uk.}
	}
	\date{September 9, 2021\\
	\vspace{5pt}
		{\normalsize To appear in:\\ {\em Online and Matching-Based Market Design}.\\ 
		Federico Echenique, Nicole Immorlica and Vijay V. Vazirani, Editors. Cambridge University Press.
		\textcopyright \ 2021 }
	}

\begin{document}

\maketitle

\section{Introduction}

In this chapter, we explore topics inspired by two peculiarities present in a number of matching markets and market design problems.

The first one concerns foundations of choice behavior in matching markets in the tradition of revealed preference. 
Agents in such markets face \emph{combinatorial} choice problems.
The hallmark of such problems is the agent's freedom to combine elementary things into a bundle that determines the outcome.
For instance, in labor matching markets, the firms can freely choose any combination of potential employees who have applied (including choosing none of them).

The second one concerns the viability of designing a mechanism through a two-part factorization: the particulars of how the planner and the agents (or their algorithmic proxies) interact in the mechanism, and the planner's preferences and objectives.
Recent developments in object allocation and matching demonstrate that the planner's choice behavior in the mechanism can come from a rich set of possibilities without disturbing strategic incentives of agents.

The salient example is public school choice through centralized assignment, which resembles a two-sided matching market, but with an important distinction.
Students are agents both in the sense of being able to determine their own choices and in the sense of having normative welfare weight for the planner.
However, schools are typically considered a resource for consumption and so lack agency in the first sense.
From Chapter ``Two-sided markets", we know that the deferred acceptance algorithm with students as proposers yields the stable outcome most preferred by students.
We also know that the direct mechanism that determines an outcome using it is strategy-proof for students.

But if schools lack agency, how can they make choices in the algorithm?
The insight is that the planner can assume the role of agent of the schools in the operation of the deferred acceptance algorithm, making choices amongst the applicants at each step.
These choice can vary according to the planner's preferences or objectives while preserving incentive-compatibility if the choice rule is independent of preference reports and satisfies additional conditions.

So how much can these choice rules vary?
The important requirement is that rejections made in the course of the deferred acceptance algorithm are final in a certain sense.
This property holds true if the choices are made according to maximization of a fixed substitutable preference.
What is relevant for design is that a growing list of interesting choice rules not based on preference maximization have been defined that also ensure the finality of rejections property.

What is stability if schools do not have preferences?
A straightforward interpretation is that stability of an matching is the absence of a school and a student (or group of students) who prefers it to their assignment, where the school would take them given its current assignees when the choice is according to the planner's choice rule.
Unlike in truly two-sided markets, however, we cannot be assured of Pareto efficiency of a stable matching since only students have welfare relevance and the requirement of stability places constraints on the allocation.
Any changes to a stable allocation that would be chosen according to the choice rule would also reduce the welfare of some students.
If the choice rule is a complete expression of various aspects of planner objectives separate from student welfare, then stable allocations necessarily exhibit trade-offs between the collective welfare of students and every aspect of planner objectives, whereas allocations that are ``blocked'' necessarily have an aspect without a trade-off.

The upshot is that incentive-compatible market design for combinatorial or discrete goods may fruitfully be pursued through \emph{choice function design}, especially when Pareto efficiency is not paramount but instead subject to being traded-off with other planner objectives.
Our aim is to provide the foundation necessary to follow this recent design paradigm.

\medskip

Section \ref{choice:sec:modelChoice} defines the primary choice models of interest and discusses issues that pertain to modeling choice behavior, including a detailed discussion of combinatorial choice.
Section \ref{choice:sec:revPref} covers the minimum amount of revealed preference theory necessary to understand some of its implications for combinatorial choice.
Section \ref{choice:sec:combChoice} is centered on behavioral properties of combinatorial choice, while 
Section \ref{choice:sec:PIchoicefunction} describes the remarkable structure induced by the property of path independence.
In Section \ref{choice:sec:choiceRules} we study some combinatorial choice rules featured in recent market design applications. 
In Section \ref{choice:sec:choiceAndDA} formalizes the discussion above and illustrates the relationship between choice behavior, stability, and deferred acceptance.
For students seeking a deeper understanding, we offer some guidance in Section \ref{choice:sec:notes} and exercises in Section \ref{choice:sec:exercises}.

\section{Modeling Choice Behavior}
\label{choice:sec:modelChoice}

We study \emph{choice behavior} of an \emph{actor} in an \emph{environment}.
The essential description of any \emph{instance} of choice behavior of the actor in an environment consists of a \emph{choice} (or \emph{chosen alternative}), the \emph{budget set} of possible \emph{alternatives} from which the choice was made, and the prevailing \emph{conditions} in the environment.

A collection of instances of choice behavior may be factorable into a \emph{choice setting} and a \emph{choice function}, where
the choice setting defines a \emph{space of alternatives}, a \emph{domain of states}, and a \emph{budget map} identifying the budget set at each state,\footnote{
In this framework, the state encodes the environmental conditions and the budget set, which is decoded using the budget map.
No two instances could have the same state but different budget sets.
}
and
the choice function identifies the choice at each state.\footnote{
If two instances have the same state but different choices,
then choice behavior is not deterministic with respect to the state.
We can relax the requirement that the choice function identify only one alternative as chosen.
Why might there be two instances with the same state and different choices, given that the budget set is completely determined by the state?
This suggests an incompleteness in the description of the state arising from two sources, the environmental conditions or the actor's internal state.
}

A \emph{choice model} is a choice setting together with a \emph{choice function space}, which describes the class of \emph{admissible} choice behaviors.

\subsection{General Model of Choice Behavior}

A \df{general choice model} is a tuple $(\mcX, \Pi, B, \mbbC)$, where
$\mcX$ is a space of alternatives,
$\Pi$ is a domain of states,
$B$ is a \emph{budget map} taking each state $\pi$ to a nonempty budget set $B(\pi) \subset \mcX$,
and $\mbbC$ is a space of choice functions, which are maps $c$ from $\Pi$ to $2^\mcX$ such that $c(\pi)$ is a nonempty subset of $B(\pi)$.
An example of a problem that calls for a general choice model is the study of demand, where budget sets are defined as a function of prices and an expenditure limit.

When the environmental state is thought to affect choice behavior only through the budget set it defines, we can simplify to a \df{pure choice model}.
This is a tuple $(\mcX, \mc B, \mbbC)$, where
the budget set map is replaced by a budget set domain, $\mc B$, which is a collection of subsets of $\mcX$,
and
choice functions are maps $c$ from $\mc B$ to $2^\mcX$ such that $c(B)$ is a nonempty subset of $B$.

We make a few important remarks on interpretation.
\begin{remark}[Decisiveness]
We interpret choice behavior as \emph{decisive} if and only if the chosen set from a budget set is a singleton.
This is to say that alternatives are \emph{mutually exclusive} of each other.
There are a few different interpretations of non-singleton choice, which we simply call \emph{partial} choice.
Because formal axioms require interpretation, it is important to lay out the interpretation for any formal analysis.
A common interpretation is indifference of the chooser between alternatives in the chosen set.
Another interpretation is that the chosen set is an intermediate stage in a multi-stage choice procedure, with the ultimate alternative determined at a later stage.\footnote{
This interpretation might seem to contradict the mutual exclusivity of alternatives.
But this can modeled as partially observed choice, where the final chosen alternative is only known to be in the partially observed chosen set.
}
\end{remark}

The usefulness of these models for understanding positive or normative theories of choice depends upon whether there is more structure to the entities of the model, that is, the space of alternatives and the domain of states and budget sets.
The assumptions made about structure helps classify a variety of settings and models of choice.

\subsubsection{Structure on the set of alternatives}

In an \emph{abstract setting}, the space of alternatives $\mcX$ is simply a set with no structure assumed.
It is the baseline model for social choice theory, examined in Chapter ``Objectives".
Given general applicability of results for the model, it is the setting for the standard theory of rationalizability and revealed preference.

In an \emph{economic setting}, the nature of the goods to be chosen define the structure of the space of alternatives, where goods are construed broadly to include immaterial things like relationship matches or radio spectrum.
Let $K$ be the number of kinds of goods, with each kind being homogeneous and measurable.\footnote{
For our purposes, assuming homogeneity is without loss of generality, as long as the notion of indistinguishability defines an equivalence relation over goods.  Heterogeneity within a kind of good can be dealt with by an appropriate homogeneous refinement of this kind.}
Then we can model $\mcX$ as a subset of the vector space $\R^K$, such as the positive orthant $\R_+^K$, with its usual order and algebraic structure.
Each alternative is a \emph{bundle} of the $K$ kinds of goods, represented by a vector in $\mcX$.
Discreteness of kind $k$ is represented when the $k$-th component of each vector in $\mcX$ is an integer,
whereas perfect divisibility is represented by the $k$-th component varying in an interval of real numbers.\footnote{
For $a \in \R$, $x = (x_k)_{k \in K}$, and $y = (y_k)_{k \in K}$ in $\R^K$, $x \leq y$ if and only if $x_k \leq y_k$ for each $k \in K$, $x + y$ is the tuple $(x_k + y_k)_{k \in K} \in \R^K$, and $ax$ is the tuple $(ax_k)_{k \in K} \in \R^K$.
So, the space of alternatives $\mcX$ is a subset of an ordered vector space.
If $\R$ is replaced by $\Z$ throughout, then the background space for $\mcX$ is an ordered $\Z$-module.
}
For example, if all $K$ kinds of goods come in non-negative discrete quantities, we have a \emph{discrete setting} with $\mcX$ a (rectangular) subset of $\Z_+^K$.

The quintessential discrete setting is one with heterogeneous and perfectly indivisible goods.
Take, for example, a firm looking to hire workers from a labor market, where the market conditions affect the pool of workers available for hire.
A resolution of the choice problem is a cohort of hired workers.
Or consider, instead, an MBA student looking to determine courses to take in a given semester.
A resolution of the choice problem is a schedule of courses.
Here, $\mcX$ is said to have a \df{combinatorial goods}\label{choice:df:combGoodsStruct} structure, since alternatives in $\mcX$ involve combinations of the indivisible goods.
Market design and matching is replete with combinatorial choice problems.
If $\X$ is the set of indivisible goods, each alternative is a bundle that combines zero or one unit of each good, representable as a subset of $\X$.
The space of alternatives $\mcX$ is the set $2^\X$, which is the power set of $\X$ equipped with its canonical structure as a Boolean algebra.\footnote{
A collection of subsets of a given set is a Boolean algebra if arbitrary unions, as well as arbitrary intersections, of members of the collection are also members of the collection, and complements of members of the collection are also members of the collection.
In fact, $\mcX$ as defined has the structure of a power set algebra, which has a minor distinction as a special type of Boolean algebra.
By the representation theorem for Boolean algebras, a finite Boolean algebra is isomorphic to some power set algebra, but there are infinite Boolean algebras that do not have a power set algebra representation.
We shy away from the realm of the infinite in this chapter, so the distinction is immaterial.
}

\begin{remark}\label{choice:rmk:elementsAreFeatures}
There is an important distinction between an element of $\X$ and the singleton set containing this element.
Each element can be interpreted as a distinct ``feature'' or dimension of the space of alternatives, whereas the singleton set is the unique alternative that has as its only feature the element it contains.
\end{remark}
\begin{remark}
The alternative with none of the elements of $\X$ present is represented by the empty set.
In the matching market examples above, 
it is the alternative of hiring no worker or of not enrolling in school for the semester.
\end{remark}

\subsubsection{Structure on the domain of states and budget sets}

In economic settings with goods having market prices, the state of the environment comprises a price vector quoted in terms of a unit of account, called ``money'', and an endowment of money or of goods, and a demand function describes the choice behavior of the agent.
This is a setting with \df{price-based budget sets}.\footnote{
A price-based budget set is defined by an inequality constraint that is linear in the prices ($\dotprod{p}{\x} \leq b$), with the bound $b$ a given expenditure limit (money endowment) or equal to the market value $\dotprod{p}{\x_0}$ of a given endowment of goods $\x_0$.}
\label{choice:df:priceBasedBudgets}

In some cases, money is a currency that has valuable use outside of the context of the problem domain.
Examples include the auction environments used to sell bands of electromagnetic spectrum (see Chapter ``Spectrum Auctions") or advertisement slots on a webpage (see Chapter ``Online Matching in Advertisement Auctions").
However, in some other cases, money is solely a construct of a mechanism or setting, with no valuable use outside of this context.
Examples of this include the allocation of courses in some business schools through pseudomarket mechanisms (see Chapter ``Pseudomarkets").
So, money with outside value is a good that directly affects the agent's utility, whereas money without outside value is artificial or \emph{token} money that only indirectly affects the agent's utility through its use as a unit of account for prices and expenditure limits in the budget constraint.

In an abstract setting of pure choice, the domain of states is identified with a domain of budget sets.
Budget domains may yet have structure in a neutral setting, through conditions on cardinality.
The budget domain is \df{complete}\label{choice:df:completeBudgetDomain} if it contains every finite subset of alternatives.
It is \df{additive}\label{choice:df:additiveBudgetDomain} if the union of two budget sets in the domain is also in the domain.
It is a \df{connected}\label{choice:df:connectedBudgetDomain} domain if for every trio of alternatives (possibly indistinct) there is a budget set in the domain that contains exactly these alternatives.

With a structure on the space of alternatives, natural structures on budget sets and domains emerge.
For example, in economic settings with goods in measurable quantities, \emph{free disposal} allows an actor to reduce the quantity of one or more kinds of goods in a bundle without cost.
In matching settings, where a good represents a particular relationship match, \emph{voluntary participation} allows an actor to drop one or more relationships without dropping others.
Since for these settings we represent alternatives as vectors in an ordered vector space with partial order $\leq$, free disposal or voluntary participation maps to the requirement that a budget set is \emph{downward closed}, which means that it contains all vectors that are lesser in the partial order than some vector in it.\footnote{In formal notation, budget set $B$ is said to be downward closed if for every $x, y \in \mcX$, if $x \in B$ and $\y \leq x$, then $y \in B$.}
We say the budget domain is \df{comprehensive}\label{choice:df:comprehensiveBudgetDomain} if every budget set in it is downward closed.
For example, the price-based budgets domain is comprehensive.

In many matching settings, in addition to budgets being downward closed, actors possess another liberty in constructing a bundle, that of unrestricted combination.
Suppose a firm has applicant pools $\bY$ and $\bY'$ from two different recruiting channels.
What are the bundles (that is, teams of workers) the firm could consider hiring?
Voluntary participation permits any team $\bZ$, drawn from one of the two pools.
Unrestricted combination means the firm could form a team from any combination of the applicants in pool $\bY$ and in pool $\bY'$.
Since bundles are subsets of $\X$, combinations are defined by unions.

\begin{definition}
A budget domain is a \idf{combinatorial choice domain}{combinatorial choice domain}\label{ch:choice:combChoiceDomain} if
\begin{inparaenum}[(1)]
\item the domain is \df{comprehensive},
\item every budget set $B$ in the domain is \emph{join closed}, that is, for every pair of bundles $x, y \in B$, $x \vee y \in B$,
\item every bundle $\bY$ is \emph{potentially budget constrained}, that is, $\bY$ is $\subset$-maximal for some budget in the domain.
\end{inparaenum}
\end{definition}
If a bundle $\bY$ is potentially budget constrained, there is a price vector $p$ and expenditure limit $b$ so that $\bY$ is inclusion maximal in a price-based budget set.

\subsection{Combinatorial Models of Choice Behavior}
\label{choice:sec:combModel}
As our overview has made clear, the combinatorial setting brings a lot of structure with it.
We define two models of choice that directly incorporate this structure.

A \idf{combinatorial choice model}{combinatorial choice model} is a tuple $(\X, \mc D, \mbfC)$, where $\X$ is a finite set of \df{elements}, $\mc D \subset 2^\X$ is a nonempty domain of \emph{option sets}, and $\mbfC$ is a set of \idf{combinatorial choice functions}{combinatorial choice function}, which are functions $C: \mc D \to 2^\X$ such that $C(\bY) \subset \bY$ for each option set $\bY$.
We will generally assume that the domain of option sets is \emph{complete}, that is $\mc D = 2^\X$, in which case we drop its notation in the tuple.

\begin{remark}
Each $C \in \mbfC$ models \emph{decisive} choice behavior, even when, for some option set $\bY$, $C(\bY)$ is a set with more than one element.
To reiterate Remark \ref{choice:rmk:elementsAreFeatures}, elements are features or dimensions of the space of alternatives, and not themselves alternatives.
Instead, an alternative or bundle is a set of elements, hence a subset of $\X$.
To allow for partial choice, choice correspondences from $\mc D$ to $2^\X$ are needed.
We do not pursue this more general approach.

\end{remark}
\begin{remark}
An option set is not a budget set, because it is not the set of all bundles available.
Instead, it comprises the elements that may be combined into bundles.
However, when the domain of budgets is infinite, as in the next model we discuss, correspondences are largely unavoidable when some sort of continuity is desirable.
\end{remark}

A \idf{combinatorial demand model}{combinatorial demand model} is a tuple $(\X, \Omega, B, \mbfD)$, where $\X$ is a finite set of elements, $\Omega$ is a nonempty set, $B$ is the budget map $B: \R_{++}^\X \times \Omega \toto 2^\X$, and $\mbfD$ is the set of \idf{combinatorial demand correspondences}{combinatorial demand correspondence}, which are correspondences $D: \R_{++}^\X \times \Omega \toto 2^\X$ such that $D(p, \omega) \subset B(p, \omega)$ for each $(p, \omega) \in \R_{++}^\X \times \Omega$.

As with the combinatorial choice model above, an alternative is a bundle of goods that is defined as a subset of elements from $\X$.
Identify each bundle $\bZ \in 2^\X$ with the vector in $\R^\X$, also denoted by $\bZ$, that has an entry of $1$ at index $a \in \X$ if $a \in \bZ$ and an entry of $0$ otherwise.
We are most interested in settings with a money endowment and linear prices, and so take $\Omega$ to be a subset of $\R_+$ and $B(p, \omega) = \set{\bZ \in 2^\X: \dotprod{p}{\bZ} \leq \omega}$.
It should be clear that this is simply a general choice model $(\mcX, \Pi, B, \mbfC)$ with combinatorial goods (see p.\pageref{choice:df:combGoodsStruct}) and price-based budget sets (see p.\pageref{choice:df:priceBasedBudgets}), where $\mcX = 2^\X$ and $\Pi$ has been factored in to a price vector space $\R_{++}^\X$ and an endowment space $\Omega$.

\subsection{Faithful Representations of Combinatorial Choice Models}\label{ch:choice:faithfulMap}
It is natural to ask what is the relationship between combinatorial choice models and the pure choice models previously described.
As we now describe, each combinatorial choice model is \emph{behaviorally isomorphic} to some pure choice model with combinatorial goods and budget sets.
This means that each combinatorial choice model is a \emph{faithful} representation of a pure choice model.

Given a combinatorial choice model $(\X, \mc D, \mbfC)$,
define $\mcX = 2^\X$, 
$\mc B = \set{2^\bY: \bY \in \mc D}$, 
and $\mbbC$ to be the set of all $c:\mc B \to 2^\mcX$ such that $c(2^\bY) = \set{C(\bY)}$ for every $\bY \in \mc D$.
Then $(\mcX, \mc B, \mbbC)$ is a pure choice model where
the space of alternatives $\mcX$ is structured as a powerset algebra, 
the budget domain $\mc B$ is both comprehensive and join-closed,
and
the choice functions in $\mbbC$ are decisive.
Moreover, if the domain of option sets is complete ($\mc D = 2^\X$), then $\mc B$ is a combinatorial choice domain.
The mapping from combinatorial choice models to pure choice models so defined is denoted $\mathfrak F$.

Take a pure choice model $(\mcX, \mc B, \mbbC)$ where
all choice functions in $\mbbC$ are decisive,
the space of alternatives is a finite Boolean lattice, $(\mcX, \leq, \vee, \wedge)$,
and $\mc B$ is a collection of subsets of $\mcX$ that is comprehensive with respect to $\leq$ and join-closed with respect to $\vee$.
Denote the mapping defined below from pure choice models with this structure to combinatorial choice models by $\mathfrak G$.
Let $\X$ be the set of \emph{atoms} of the Boolean lattice $\mcX$.
Atoms are simply those members of the lattice with only the bottom of the lattice below them.
These correspond naturally to the set of elements in a combinatorial model, because every member of the lattice is the join of the set of atoms below it.
By the representation theorem for Boolean lattices, there is an isomorphism from $\mcX$ to the powerset algebra $2^\X$ of atoms, with the usual set operators and inclusion serving the lattice operators and order.\footnote{
Representability of $\mcX$ as a powerset algebra is characterized by $\mcX$ being a complete and atomic Boolean lattice, which are conditions automatically satisfied if $\mcX$ is finite.
The isomorphism between combinatorial and pure choice models can be extended to infinite spaces if we strengthen join-closed to complete-join-closed.
}
That is, $\mathfrak G$ maps each $\x \in \mcX$ to $\set{a \in \X: a \leq \x}$.
Each budget set $B \in \mc B$ is mapped to an option set $\mathfrak G(B) = \bigcup_{\y \in \mc B} \mathfrak G(\y)$.
Since $B$ is comprehensive and join-closed, $\wedge B$ is the bundle $\x \in B$ that the greatest in terms of the underlying partial order on alternatives $\leq$.
Then $\mathfrak G(\x)$, the set of atoms below $\x$, is equal to $\mathfrak G(B)$.
So, define the domain of option sets $\mc D = \set{\mathfrak G(B) : B \in \mc B}$, which can be seen as a collection of subsets of $\X$.
Moreover, if $\mc B$ is a combinatorial choice domain, then the domain of option sets is complete.
Finally, each choice function $c \in \mbbC$ is mapped to a function $\mathfrak G(c) = C: \mc D \to 2^\X$ such that $C(\mathfrak G(B)) = \mathfrak G(c(B))$.
It should be clear that each type of object is mapped bijectively from one model to the other.

\section{Revealed Preference and Choice Behavior}
\label{choice:sec:revPref}

A \emph{theory of choice} describes how choice behavior is determined.
It will posit the \emph{existence} of theoretical entities (e.g. preferences, priorities, information) and rules or laws of how these entities produce choice behavior.
When applied to a particular choice setting, it generates a model of choice with that setting.
If the generated choice model describes the choice behavior being studied, we say the behavior (or model) is \emph{rationalized} by the theory.

In this section we study rationalizability by the theory of preference-based (utility-based) choice, discussed in Chapter ``Objectives". 
Recall that the theory postulates that the agent has a preference relation (utility function) and as a rule chooses any one of the most preferred (utility-maximizing) alternatives from a given budget set.

Fix a pure choice model $(\mcX, \Pi, B, \mbbC)$.

We model preferences by a binary relation $\wpref$ on $\mcX$, where $x \wpref y$ denotes ``$x$ is at least as preferred as $y$''.
Let $\overline{\mc R}$ be the set of all preference relations on $\mcX$.
Unlike the definition in Chapter ``Objectives'', we allow for intransitivity and incompleteness.
The transitive closure of a binary relation $R$, denoted $\tau(R)$, 
is the inclusion-smallest transitive relation that contains $R$.
It means that $x \mr \tau(R) y$ if and only if there exists a finite sequence $x_0, \dots, x_n$ such that $x_0 = x$, $x_n \mr R y$ and for every $m = 0, \dots, n - 1$, $x_m \mr R x_{m+1}$.\footnote{
The transitive closure of $R$ is equivalently defined by $\tau(R) = \bigcap_{R' \in \mc R^\tau} R'$, where $\mc R^\tau$ is the set of all transitive relations $R'$ such that $R' \supseteq R$.
}
The set of $\wpref$-\emph{greatest} alternatives in a subset $B$ of $\mcX$, denoted $\Top(B, \wpref)$, is equal to $\set{x \in B: \forall y \in B, x \wpref y}$.
The \df{preference-maximization choice rule} $\Gamma$ determines choice to be the set of preference-greatest alternatives in the budget set.
It maps each preference relation $\wpref$ to a choice function $\Gamma^\wpref$, defined by $\Gamma^\wpref(B(\pi)) = \Top(B(\pi), \wpref)$ for every $\pi \in \Pi$.
We say $\Gamma^\wpref$ is the choice function induced by $\wpref$.

A choice function $c \in \mbbC$ is (transitively) \df{rationalized} by a (transitive) preference relation $\wpref$, which is called a (transitive) \emph{rationalization} of $c$, if $c$ is induced by $\wpref$ under the preference-maximization choice rule, that is, if $c = \Gamma^\wpref$.

\subsection{Rationalizability and Revealed Preference}

Questions for any theory intended as a positive analysis of observed choice include the following:
Is the theory \emph{falsifiable}, that is, are there observable choice patterns not consistent with the theory?
Is the theory \emph{testable}, that is, given data on choice behavior, is there an effective test based solely on the data that will correctly falsify the theory or correctly provide a rationalization?
We address both these questions through an analysis of revealed preference.

We say that alternative $x$ is \idf{revealed preferred}{revealed preference} to alternative $y$, denoted $x \mr R_{c} y$,
if there exists a problem $\pi \in \Pi$ such that $x$ is chosen and $y$ is in the budget set (i.e. $x \in c(\pi)$ and $y \in B(\pi)$).
If for some problem $x$ is chosen and $y$ is in the budget set but not chosen, we say that $x$ is \idf{revealed strictly preferred}{revealed strict preference} to $y$, denoted $x \mr R^s_{c} y$.
It is important to note that $R^s_c$ may not be the asymmetric part of $R_c$.
For example, there could be one budget set at which $x$ and $y$ are chosen (so $x \mr R_c y \mr R_c x$) and another budget set at which $x$ and $y$ are available but only $x$ is chosen (so $x \mr R^s_c y$).

We begin with a result that explains the central place of the revealed preference relation in the analysis of rational choice.
From the point of view of testability, it implies that a potentially exhaustive search for a rationalization in the preference space is not required.
\begin{proposition}\label{prop:revPrefRat}
A choice function $c$ is rationalizable if and only if it is rationalized by its revealed preference relation $R_c$.
\end{proposition}
\begin{proof}
The ``if'' direction is immediate.
To prove the other direction, notice that $c(\pi) \subset \Top(B(\pi), R_c)$ follows directly from the definition of revealed preference.
Moreover, any rationalization must extend the revealed preference relation.
That is, for every pair of alternatives $x$ and $y$ and every rationalization $\wpref$ of $c$, $x \mr R_c y$ implies $x \wpref y$.
Since $x$ revealed preferred to $y$ implies there is a $\pi \in \Pi$ such that $x \in c(\pi)$ and $y \in B(\pi)$, the definition of rationalization implies $c(\pi) = \Top(B(\pi), \wpref)$, and so $x \wpref y$.
Finally, the set of greatest alternatives expands when the preference relation is extended, that is, $\Top(B(\pi), R_c) \subset \Top(B(\pi), \wpref)$.
\end{proof}

\subsection{WARP and Rationalizability} 

A choice function $c$ satisfies the \idf{weak axiom of revealed preference (WARP)}{weak axiom of revealed preference} if for every pair of alternatives $x$ and $y$, if $x$ is revealed preferred to $y$, then $y$ is not revealed strictly preferred to $x$.

The weak axiom of revealed preference is the foundational axiom of revealed preference theory.
We begin exploration of its implications in the pure choice setting.
\begin{theorem}\label{thm:WARPimplies}
If $c$ is a choice function that satisfies WARP, then 
\begin{enumerate}
	\item it is rationalizable.
	\item it is transitively rationalizable if the domain of budget sets $B(\Pi)$ is connected or additive.
\end{enumerate}
\end{theorem}

\begin{proof}
From the proof of Proposition \ref{prop:revPrefRat}, we simply need to show that WARP implies, for every $\pi \in \Pi$, $\Top(B(\pi), R_c) \subset c(\pi)$.
The reader should verify that a choice function satisfies WARP if and only if the revealed strict preference relation is equal to the asymmetric component of the revealed preference relation.

As for rationalizability by a transitive preference relation, suppose first that the domain of budget sets is connected.
Let $x \mr R_c y $ and $y \mr R_c z$.
Take the set $\set{x, y, z}$, which connectedness assures us is a budget set in the domain.
By WARP, we see that $z$ chosen implies $y$ chosen and in turn so $x$ chosen.
Indeed, this argument symmetrically applies to all members of $\set{x, y, z}$. 
Then, assuming choice is non-empty, all three are chosen, implying $x \mr R_c z$ and so transitivity of the revealed preference relation.

The case with an additive domain of budgets is left as Exercise \ref{ex:WARPimplies}.
\end{proof}

\smallskip

We can understand the implication of WARP for combinatorial choice models by making use of the isomorphisms mapping between combinatorial and pure choice.

\begin{theorem}\label{thm:WARPequivIRE}
Suppose a combinatorial choice model $(\X, 2^\X, \mbfC)$ and a pure choice model $(\mcX, \mc B, \mbbC)$ are isomorphic, i.e. $(\X, 2^\X, \mbfC)$ and $(\mcX, \mc B, \mbbC)$ are isomorphically mapped from the first to the second by $\mathfrak F$ and second to the first by its inverse $\mathfrak G$.
Then $C \in \mbfC$ satisfies IRE\footnote{See p.\pageref{ch:choice:def:IRE} for the definition.} if and only if $c = \mathfrak F(C)$ satisfies WARP.
\end{theorem}
The behavioral implications for a combinatorial choice model of various conditions formalized in the pure choice model can be obtained through its faithful representation.
See Section \ref{choice:sec:notes} for a reference.

\smallskip
We turn next to the combinatorial demand model.
The utility maximization theory posits that the agent has a utility function $u$ and his demand $D^u(p, m)$ is derived as the set of choices that maximize $u$ subject to constraints placed by the prices $p$ and money-budget $m$.
The appropriate analysis depends on whether money provides direct utility.
We assume so.
The agent gets utility $u(A, t)$ from consuming a pair of a bundle of items and a quantity of money $(A, t)$, where $u$ is strictly monotone in money.

Faced with $(p, m) \in \Pi$, where $m$ could be interpreted as an endowment of money, the agent chooses $(A, t)$ to maximize $u(A, t)$ subject to the budget constraint $\dotprod{p}{A} + t \leq m$.
Notice that money is the numeraire since its price equals 1.
The derived demand for bundles of items $D^u(p, m)$ consists of each bundle $A \in \mcX$ that for some $t \in \R$ maximizes $u$ at $(p, m)$.
This is because we can identify $t$ from $A \in D^u(p, m)$ by $t = m - \dotprod{p}{A}$, given monotonicity in money.

Consider utility that is quasilinear in money, as introduced in Chapter ``Objectives", so that $u(A, t) = v(A) + t$ for some $v:2^\X \to \R$ called the valuation function.
As long as any lower bound on money consumption is not binding, quasilinearity simplifies demand analysis by eliminating income effects.
That is, for any two $(p, m)$ and $(p, m')$, $D^u(p, m) = D^u(p, m')$.

Given this irrelevance of the money-budget on the demand for items, we shall not indicate its level.
Can a given demand correspondence be explained by maximization of some objective function that is quasilinear in prices?

A positive answer to the question requires some intuitive restrictions on $D$.
First, we say $D$ satisfies \df{law of demand} if for every $p, p'$ and every $A \in D(p)$ and every $A' \in D(p')$,
$\dotprod{p - p'}{A - A'} \leq 0$.
To understand the meaning of this requirement, notice that it implies a version of the \emph{Weak Axiom of Revealed Preference} adapted to this setting.
WARP states that for any bundles $A$ and $A'$ and prices $p$ and $p'$, if bundles $A$ and $A'$ are demanded at prices $p$ and $p'$, respectively, and furthermore bundle $A'$ is worth less at prices $p$, then bundle $A$ is more expensive than $A'$ at price $p'$.
That is, if $A \in D(p)$, $A' \in D(p')$ and $\dotprod{p}{A'} < \dotprod{p}{A}$, then $\dotprod{p'}{A} > \dotprod{p'}{A'}$.

We also need a continuity property.
Demand function $D$ is \df{upper hemicontinuous} if for every $p \in \R^\X_{++}$, there exists an open neighborhood $V$ of $p$ such that $D(q) \subset D(p)$ for every $q \in V$.

\begin{proposition}
For any quasilinear utility function $u$, the derived demand function $D^u$ on domain $\Pi$ satisfies the Weak Axiom of Revealed Preference and upper hemicontinuity.
\end{proposition}
\begin{proof}
To prove WARP, let $p \in \R_{++}^\X$, $A \in D^u(p)$, $A' \in D^u(p')$ and $\dotprod{p}{A'} < \dotprod{p}{A}$.
From utility maximization, $v(A) - \dotprod{p}{A} \geq v(A') - \dotprod{p}{A'}$ and
 $v(A') - \dotprod{p'}{A'} \geq v(A) - \dotprod{p'}{A}$.
Then these two inequalities yield 
$\dotprod{p'}{A} - \dotprod{p'}{A'} \geq v(A) - v(A') \geq \dotprod{p}{A} - \dotprod{p}{A'}$.
But then $\dotprod{p}{A'} < \dotprod{p}{A}$ implies $\dotprod{p'}{A} - \dotprod{p'}{A'} \geq 0$ as desired.

Define $U(A, p) = v(A) - \dotprod{p}{A}$ for each bundle $A$ and price vector $p$.
To prove upper hemicontinuity, 
notice that for each bundle $A$, $U$ is continuous in prices.
Fix a price vector $p$.
Let $W$ be the maximum utility attained at $p$, i.e. $W = U(A,p)$ for some maximizing bundle $A \in D^u(p)$.
Define $V_A$ to be the preimage under $U(A, \cdot)$ of the open set $(\infty, W)$.
Let $V = \cap_{A \in 2^\X \sminus D^u(p)} V_A$.
Note that for every bundle $A$, $A \nin D^u(p)$ if and only if $U(A, p) < W$.
Thus, $A \nin D^u(p)$ implies $p \in V_A$, and so $p \in V$.
Continuity of a function means the preimage of an open set is open, so $V_A$ is an open set.
Then $V$ is open, since it is a finite intersection of open sets, given that $2^\X$ is finite.
Finally, for any $A \in 2^\X \sminus D^u(p)$, $V \subset V_A$, so for any $q \in V$,  $U(A, q) \in (\infty, W)$.
Thus, $A \nin D^u(q)$, establishing upper hemicontinuity.
\end{proof}

A converse also holds.
\begin{theorem}\label{thm:ratCombDemand}
(Rationalizability) For any combinatorial demand $D$ on domain $\Pi$ that satisfies the law of demand and upper hemicontinuity, there exists a quasilinear utility $u$ that rationalizes it, that is, the derived demand function $D^u$ equals $D$.
\end{theorem}
We omit a proof, a reference for which is in Section \ref{choice:sec:notes}.

\section{Combinatorial Choice Behavior}
\label{choice:sec:combChoice}

In this section, we describe some types of choice behavior that are particular to combinatorial choice settings.
Fix a combinatorial choice model $(\X, \mbfC)$ with a complete domain of option sets, and let $C$ be a combinatorial choice function.

Substitutability turns out to be essential in market design applications to ensure convergence to a stable matching when the deferred acceptance algorithm takes choice rules of institutions (such as schools) as input.  

We say $C$ satisfies \df{substitutability}\label{ch:choice:def:substitutability} if the elements chosen from a given option set that remain available in a given subset of the given option set are amongst the chosen elements from the given subset.
That is, for each $\bS, \bT \in 2^\X$,
$$\mbox{if } \bT \subset \bS,\mbox{ then } C(\bS) \cap \bT \subset C(\bT).$$
We say $C$ satisfies (combinatorial) \df{path independence}\label{ch:choice:pathIndependence} if the choice from a given option set is the same as the choice from the collection of elements chosen from each of two option sets whose union equals the given option set.
That is, for each $\bS, \bT \in 2^\X$,
$$C(\bS \cup \bT)=C(C(\bS)\cup C(\bT)).$$

In fact, path independence is equivalent to substitutability together with the following choice invariance condition.
We say $C$ satisfies \df{irrelevance of rejected elements}\label{ch:choice:def:IRE} if the removal of some rejected elements from the option set leaves the set of chosen elements unchanged.
That is, for each $\bS, \bT\in 2^\X$, $$\mbox{if }C(\bS)\subseteq \bT \subseteq \bS,\mbox{ then }C(\bS)=C(\bT).$$

\begin{theorem}
	\label{thm:PI=SubsIRE}
	A \ccf is path independent if and only if it satisfies substitutability and IRE. 
\end{theorem}

\begin{proof}
Substitutability is equivalent to the following subadditivity condition: $C(\bS \cup \bT) \subset C(\bS) \cup C(\bT)$ for every $\bS, \bT \in 2^\X$.
We will make use of this result, whose proof is left as an exercise (see Exercise \ref{choice:ex:equivSubs}).

Suppose that $C$ satisfies substitutability and IRE.
Let $\bS_1, \bS_2 \in 2^\X$.
Then
$
C(\bS_1 \cup \bS_2) \subset C(\bS_1) \cup C(\bS_2) \subset \bS_1 \cup \bS_2,
$
where the first inclusion is by substitutability and its equivalence to subadditivity, and the second inclusion is by the definition of a choice function.
Then $C(\bS_1 \cup \bS_2) = C(C(\bS_1) \cup C(\bS_2))$ follows from IRE, proving path independence.

Suppose that $C$ satisfies path independence.
By path independence, $C(\bS \cup \bT) = C(C(\bS) \cup C(\bT)) \subset C(\bS) \cup C(\bT)$ so subadditivity (and thus substitutability) is satisfied, where the second equality follows from the definition of a choice function.
To show IRE, suppose that $C(\bS) \subseteq \bT \subseteq \bS$, so that $C(\bT) = C(C(\bS) \cup \bT) = C(C(\bS) \cup \bT \cup C(\bS))$.
Then $C(C(\bS) \cup \bT \cup C(\bS)) = C(C(C(\bS) \cup \bT) \cup C(\bS)) = C(C(\bT) \cup C(\bS)) = C(\bT \cup \bS) = C(\bS)$, where the first and third equalities are from path independence and the second one from the previous line of equalities.
Putting the chain of equalities together yields $C(\bT) = C(\bS)$.
\end{proof}

We say $C$ satisfies \df{size monotonicity}\label{ch:choice:def:sizeMonotonicity} if the number of chosen elements does not decrease when the the option set is expanded.
That is, for each $\bS, \bT\in 2^\X$, $$\mbox{if } \bT \subset \bS, \mbox{ then } \card{C(\bT)} \leq \card{C(\bS)}.$$

\begin{proposition}
If $C$ satisfies substitutability and size monotonicity, then it also satisfies path independence.
\end{proposition}
\begin{proof}
We first show that IRE is satisfied.
Let $C(\bS) \subset \bT \subset \bS$.
By substitutability, $C(\bS) \cap T \subset C(\bT)$.
Then $C(\bS) \subset C(\bT)$.
By size monotonicity, $\card{C(\bT)} \leq \card{C(\bS)}$.
But then $C(\bS) = C(\bT)$.
So IRE is satisfied.
Then from Theorem \ref{thm:PI=SubsIRE}, we obtain path independence.
\end{proof}

\section{Path Independent Choice}
\label{choice:sec:PIchoicefunction}

As we will see in Section \ref{choice:sec:choiceAndDA}, path independence is crucial to arriving at a stable outcome via a deferred acceptance algorithm, so it is worth understanding the structure it entails.
We describe two results regarding the structure of path independent choice functions that offer insight into the previously studied lattice structure of stable matchings.

Fix a finite combinatorial choice model $(\X, \mbfC)$ with a complete domain of option sets, and fix a path independent \ccf $C$.
To simplify the discussion, assume that \emph{no element is irrelevant to $C$}, that is, for every $a \in \X$, there exists $\bS \in 2^\X$ such that $a \in C(\bS)$.
All the results in this section have a simple adaptation for $C$ that do not satisfy this condition.

\subsection{The lattice of maximal option sets of a path independent choice function}
\label{choice:sec:latticeMaxOptionSets}
We show how a path independent choice function can be represented by a particular lattice of sets.

Two option sets $\bS, \bT \in 2^\X$ are \emph{choice-equivalent} if $C(\bS) = C(\bT)$.
\begin{lemma}\label{choice:lem:intervalProperty}
For each $\bS \in 2^\X$, there exists a unique set $\bS^\sharp \in 2^\X$ such that for every $\bT \in 2^\X$, $\bT$ is choice-equivalent to $\bS$ if and only if $C(\bS) \subset \bT \subset \bS^\sharp$.
\end{lemma}
\begin{proof}
Let $\mathbf{pre}_C(\bY)$ be the pre-image of $\bY \in 2^\X$ under $C$, defined by $\mathbf{pre}_C(\bY) = \set{\bT \in 2^\X : C(\bT) = \bY}$.
Then the collection of all option sets that are choice-equivalent to $\bS$ is simply $\mathbf{pre}_C(C(\bS))$.
	
Define $\bS^\sharp = \bigcup \mathbf{pre}_C(C(\bS))$.
First, note that $C$ is \df{idempotent} since it is path independent (see Exercise \ref{exIdem} for the definition).
Second, if $\bT_1, \bT_2 \in \mathbf{pre}_C(C(\bS))$, then $C(\bT_1 \cup \bT_2) = C(C(\bT_1) \cup C(\bT_2)) = C(C(\bS) \cup C(\bS)) = C(C(\bS))$, where the first equality is from path independence and the third from idempotence.
This shows that the pre-image is closed under finite unions.
Since $\X$ is finite, it means $\bigcup \mathbf{pre}_C(C(\bS)) \in \mathbf{pre}_C(C(\bS))$ and so $\bS^\sharp$ is the unique maximal member of the pre-image of $C(\bS)$.
Next, if $C(\bT) = C(\bS)$, then $C(\bS) \subset \bT$, since $C(\bT) \subset \bT$ by definition.
Finally, if $\bT$ satisfies $C(\bS) \subset \bT \subset \bS^\sharp$, then $C(\bT) = C(\bS)$, since $C(\bS^\sharp) = C(\bS)$ and since $C$ satisfies IRE, given Theorem \ref{thm:PI=SubsIRE}.
\end{proof}

We say an option set $\bS$ is \emph{maximal} if there is no larger option set from which the same set of elements is chosen.
So $\bS$ is maximal if and only if $\bT\supseteq \bS$ implies $C(\bT) = C(\bS)$.
From the previous lemma we know that option set $\bS$ is maximal if and only if $\bS = \bS^\sharp$, with $\bS^\sharp$ as defined in the lemma's statement.

Let $\mathcal{M}$ denote the set of maximal option sets for $C$, ordered by set inclusion.
The following characterizes the predecessors in $\mc M$ of a maximal option set $\bS$.
\begin{lemma} \label{maximal has maximal kids}
For each $\bS \in \mathcal{M}$ and each $a\in C(\bS)$, $\bS \setminus \{a\} \in \mathcal{M}$.
\end{lemma}
\begin{proof}
To obtain a contradiction, suppose there exist $\bS\in \mathcal{M}$ and  $a\in C(\bS)$ such that $\bS\setminus \{a\}\notin \mathcal{M}$.
Let $\bS' \in \mathcal{M}$ such that $C(\bS')=C(\bS\setminus \{a\})$.
Since $\bS\setminus \{a\}\notin \mathcal{M}$, $\bS\setminus \{a \} \subsetneq \bS'$.
Now, consider the option set $\bS'\cup \{a\}$.
Note that $\bS \subsetneq \bS' \cup \{a\}$.  
Moreover, since $C$ is path independent,  $C(\bS'\cup \{a\})=C(C(\bS')\cup \{a\})$.
Since  $C(\bS')=C(\bS\setminus \{a\})$, we get $C(\bS'\cup \{a\})=C(C(\bS\setminus \{a\})\cup \{a\})$.
Again by path independence, $C(C(\bS\setminus \{a\})\cup \{a\})=C(\bS)$.
Thus, we get $C(\bS'\cup \{a\})=C(\bS)$.
Since $\bS \subsetneq \bS' \cup \{a\}$, this contradicts that $\bS\in \mathcal{M}$.
\end{proof}

We now prove that the intersection of maximal option sets is a maximal option set.
\begin{lemma}
The family of maximal sets $\mc M$ is intersection-closed, that is, $\bS_1, \bS_2 \in \mc M$ implies $\bS_1 \cap \bS_2 \in \mc M$.
\end{lemma}
\begin{proof}
Let $\bS_1, \bS_2 \in \mc M$ and define $\bT = \bS_1 \cap \bS_2$.
To show $\bT \in \mc M$, it is sufficient to show that $\bT^\sharp = \bT$.
Notice that $C(\bS_1 \cup \bT^\sharp) = C(C(\bS_1) \cup C(\bT^\sharp)) = C(C(\bS_1) \cup C(\bT)) = C(\bS_1 \cup \bT) = C(\bS_1)$, where the first and third equality are from path independence and the second is from the definition of choice-equivalence.
Then $\bS_1 \cup \bT^\sharp$ is choice-equivalent to $\bS_1$.
By the same argument, $\bS_2 \cup \bT^\sharp$ is choice-equivalent to $\bS_2$.
Since $\bS_1$ and $\bS_2$ are maximal, $\bT^\sharp$ is a subset of $\bS_1$ and $\bS_2$, but by Lemma \ref{choice:lem:intervalProperty}, $\bT^\sharp \supseteq \bT = \bS_1 \cap \bS_2$, proving $\bT^\sharp = \bT \in \mc M$.
\end{proof}

Since $\X$ is a clearly a maximal set, $\mc M$ has a top, that is, a greatest member $\X$ (under $\subset$).
A useful fact about any finite and intersection-closed family of sets with a top is that set inclusion is a complete lattice order, with the meet operator coinciding with set intersection.\footnote{
The join operator is only equal to the set union if the family is also closed under set unions.
Otherwise, the join of two sets is given by the intersection of all upper bounds in the family of the two sets.}

\begin{theorem}\label{choice:thm:closureLattice}
The family of maximal sets $\mc M$ ordered by set inclusion is a complete lattice of sets, with $\X$ at the top and $\emptyset$ at the bottom.
\end{theorem}

Also observe that all maximal sets and the entire choice lattice can be constructed starting from $\X$ and subtracting a chosen element at each step.

\begin{example}
\label{choice:eg:closurelattice} 

Suppose $\X=\{a,b,c\}$ and suppose $C$ satisfies
$C(\{a,b,c\})=\{a,b\}$ and $C(\bS)=\bS$ for each $\bS$ with $\card{S} \leq 2$.
It is easy to verify that $C$ satisfies path independence.
Figure \ref{lattice_figure} displays the Hasse diagram of the lattice of $\mc M$ defined by $C$.
To see how it is obtained, first imagine drawing the top of the lattice, $\X=\{a,b,c\}$.
Then, by removing each of the chosen elements $b$ and $a$, we obtain the predecessors of $\X$ a level below: $\set{a, c}$ and $\set{a, c}$.
The remainder of the diagram is obtained in a similar fashion.

	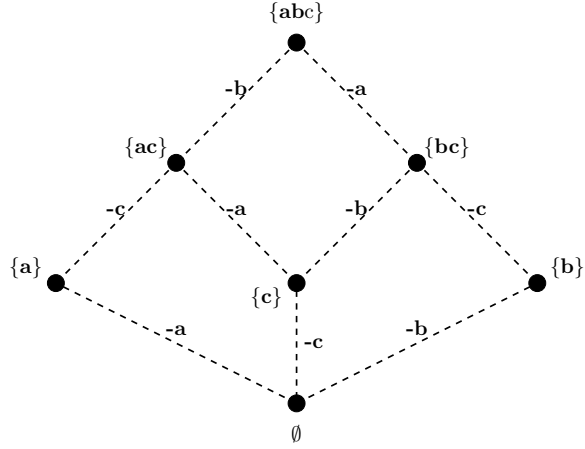
\begin{figure}
		\label{lattice_figure}
		
		\begin{center}
		\scalebox{0.8}{
			\begin{tikzpicture}
			
			
			\node(l1) at (0,4.5) {\{\textbf{ab}c\}};
			\node(1) at (0,4) {};
			\draw [fill = black](1) circle (4pt);
			
			
			\node (l21) at (-2.5,2.25) {\{\textbf{ac}\}};
			\node(21) at (-2,2) {};
			\draw [fill = black](21) circle (4pt);
			\node (l22) at (2.5,2.25) {\{\textbf{bc}\}};
			\node(22) at (2,2) {};
			\draw [fill = black](22) circle (4pt);
			
			
			\node (l31) at (-4.5,0.25) {\{\textbf{a}\}};
			\node(31) at (-4,0) {};
			\draw [fill = black](31) circle (4pt);
			\node (l32) at (-0.5,-0.25) {\{\textbf{c}\}};
			\node(32) at (0,0) {};
			\draw [fill = black](32) circle (4pt);
			\node (l33) at (4.5,0.25) {\{\textbf{b}\}};
			\node(33) at (4,0) {};
			\draw [fill = black](33) circle (4pt);
			
			
			\node (l41) at (0,-2.5) {$\emptyset$};
			\node(41) at (0,-2) {};
			\draw [fill = black](41) circle (4pt);
			
			
			\draw [dashed, color = black, thick] (1) -- (21) node[draw=none,fill=none,midway,above ] {\textcolor {black}{\textbf{-b}}};
			
			\draw [dashed, color = black, thick] (1) -- (22) node[draw=none,fill=none,midway,above ] {\textcolor {black}{\textbf{-a}}};
			
			\draw [dashed, color = black, thick] (21) -- (31) node[draw=none,fill=none,midway,above ] {\textcolor {black}{\textbf{-c}}};
			
			\draw [dashed, color = black, thick] (21) -- (32) node[draw=none,fill=none,midway,above ] {\textcolor {black}{\textbf{-a}}};

			\draw [dashed, color = black, thick] (22) -- (32) node[draw=none,fill=none,midway,above ] {\textcolor {black}{\textbf{-b}}};
			
			\draw [dashed, color = black, thick] (22) -- (33) node[draw=none,fill=none,midway,above ] {\textcolor {black}{\textbf{-c}}};
			
			\draw [dashed, color = black, thick] (31) -- (41) node[draw=none,fill=none,midway,above ] {\textcolor {black}{\textbf{-a}}};
			
			\draw [dashed, color = black, thick] (32) -- (41) node[draw=none,fill=none,midway,right ] {\textcolor {black}{\textbf{-c}}};

			\draw [dashed, color = black, thick] (33) -- (41) node[draw=none,fill=none,midway,above ] {\textcolor {black}{\textbf{-b}}};
			
			\end{tikzpicture}
		}
		\end{center}
		\caption{Hasse diagram of lattice of maximal option sets of Example \ref{choice:eg:closurelattice}, with chosen elements in boldface}
		
	\end{figure}
	
\end{example}

\subsection{Maximizer-collecting rationalization}
\label{sec:MCchoice}

Recall that for any set $\bS$ and binary relation $\succeq$, $\Top(\bS,\succeq)$ denotes the set of $\succeq$-greatest elements in $\bS$, that is, $\Top(\bS,\succeq) =\set{a \in \bS: \forall b \in \bS, a \succeq b}$.

The \df{maximizer-collecting}\label{choice:df:MCrule} choice rule $\AM$ according to a finite sequence of linear orderings $(\succeq)_1^m = (\succeq_1, \ldots ,\succeq_m)$ is defined for each $\bS \in 2^\X$ by collecting the maximizers in $\bS$ of each linear ordering in the sequence, that is,
$$\AM[(\succeq)_1^m](\bS) = \bigcup_{i\in \{1,\ldots, m\}} \Top(\bS, \succeq_{i}).$$
A choice function $C$ has a \df{maximizer-collecting (MC) rationalization}\label{choice:df:MCrational}---also known as an Aizerman-Malishevski decomposition---of size $m \in \mathbb{N}$ if there exists a finite sequence of linear orderings $(\succeq)_1^m = (\succeq_1, \ldots ,\succeq_m)$ for which the choice function defined by the MC choice rule $\AM[(\succeq)_1^m]$ equals $C$.

To illustrate an MC rationalization, let us return to the choice function $C$ in Example \ref{choice:eg:closurelattice}.
Let $\succeq_1$ and $\succeq_2$ be the linear orderings defined as $a\succeq_1 c\succeq_1 b$ and $b\succeq_2 c \succeq_2 a$.
We will observe that $C$ has an MC rationalization via $\{\succeq_1,\succeq_2\}$.
Consider the option set $\{a,b,c\}$.
Note that $a$ is the maximizer of $\succeq_1$ and $b$ is the maximizer of $\succeq_2$.
That is, collecting the maximizers yields $\{a,b\}$, which coincides with $C(\{a,b,c\})$.
Consider the option set $\{b,c\}$.
Note that $c$ is the maximizer of $\succeq_1$ and $b$ is the maximizer of $\succeq_2$.
So, collecting the maximizers yields $\set{b, c}$, which again coincides with $C(\set{b,c})$.
It can be verified that the coincidence holds for every option set.
In fact, the next result shows that any path independent choice function has an MC rationalization.

\begin{theorem}
	\label{choice:thm:MCrational}
	A choice function is path independent if and only if it has an MC rationalization.
\end{theorem}

\begin{proof}
Let $C$ be a path independent choice function.
That any MC rationalizable choice function is path independent is for Exercise \ref{choice:ex:MC}.

Define binary relation $\rightarrow$ on $\mathcal{M}$ as follows:
for each $\bS,\bT\in \mathcal{M}$, $\bS\rightarrow \bT$ if $\bS = \bT$ or $\bT = \bS\setminus \{a\}$ for some $a\in C(\bS)$.
From Lemma \ref{maximal has maximal kids}, we see that $\rightarrow$ is the predecessor relation in the partially ordered set $(\mc M, \subset)$.
So naturally $\supseteq$ is the transitive closure of $\rightarrow$.

We will construct a sequence of linear orderings and verify that $C$ has a MC rationalization via that sequence.
Remember that $\X$ and $\emptyset$ are the top and bottom members of $\mathcal{M}$, respectively.
Take any path in $\rightarrow$ that connects $\X$ to $\emptyset$, say $\X \rightarrow \bS_1\rightarrow \cdots \rightarrow \bS_k\rightarrow \emptyset$.
By definition of $\rightarrow$, note that $k=n-1$ and there exists an ordering of elements ($a_1,\ldots ,a_{n})$ such that $\{a_1\}=\X\setminus \bS_1$, $\{a_{i}\}=\bS_{i-1}\setminus \bS_{i}$ for each $i\in \{2,\ldots, n-1\}$, and $\{a_{n}\}=\bS_{n-1}$.
This defines a linear ordering $\succeq$ by $a_1\succeq \cdots \succeq a_n$.
Let $O$ be the set of all orderings obtainable in this manner, which must have finite cardinality $m$.
We show that $C$ has a MC rationalization by any sequence $(\succeq)_1^m = (\succeq_1, \ldots ,\succeq_m)$ of linear orders drawn from $O$ without replacement.

First, we show that $C(\bS)\subseteq \AM[(\succeq)_1^m](\bS)$.
Take any $\bS\in 2^\X$.
Take any $a\in C(\bS)$.
Invoking Lemma \ref{choice:lem:intervalProperty}, let $\bT = \bS^\sharp$.
Note that $\bS\subseteq \bT$.
Consider any path in $\rightarrow$ that connects $\X$ to $\emptyset$ and includes $\bT$ and $\bT\setminus \{a\}$, recognizing that at least one must exist.
Consider the linear ordering constructed for this path, say $\succeq \in (\succeq)_1^m$.
Note that $\bT\setminus \{a\}$ constitutes the set of elements that are ranked below $a$ at $\succeq$, that is, $\bT\setminus \{a\}=\{b\in \X\setminus \{a\}: a\succeq b\}$.
Since $\bS\subseteq \bT$, $a$ is the maximizer of $\succeq$ in $\bS$.

Next, we show that $\AM[(\succeq)_1^m](\bS) \subseteq C(\bS)$.
Each $a \in \AM[(\succeq)_1^m](\bS)$ is the maximizer of $\succeq_j$ in $\bS$ for some $\succeq_j \in (\succeq)_1^m$.
Define $\bT = \X \sminus \set{b : b \succ_j a}$.
Since $a$ is the maximizer of $\succeq_j$ in $\bS$, it must be that $\bS \subseteq \bT$. 
By the construction of $\succeq_j$, $\bT \in \mc M$ and $a \in C(\bT)$.
By Theorem \ref{thm:PI=SubsIRE}, $C$ satisfies substitutability, and so $a \in C(\bS)$.
\end{proof}

Let us illustrate, using the choice function $C$ in Example \ref{choice:eg:closurelattice}, how the MC rationalization is constructed from the lattice of maximal option sets $\mc M$ as explained in the proof of Theorem \ref{choice:thm:MCrational}.
Observe that in the choice lattice in Figure \ref{lattice_figure}, there are three different paths that connect $\X=\{1,2,3\}$ to $\emptyset$:
(P1) $\{a,b,c\}\rightarrow \{b,c\} \rightarrow \{c\}\rightarrow \emptyset$, (P2) $\{a,b,c\}\rightarrow \{a,c\} \rightarrow \{a\}\rightarrow \emptyset$, and (P3) $\{a,b,c\}\rightarrow \{a,c\} \rightarrow \{c\}\rightarrow \emptyset$.
Listing the subtracted chosen elements while we follow each path yields three priority orderings: 
$a\succeq_1 b \succeq_1 c$, $b\succeq_2 c \succeq_2 a$, and $b\succeq_3 a \succeq_3 c$.
In fact, $(\succeq_1,\succeq_2,\succeq_3)$ provides an MC rationalization of $C$.
Remember that we had already discovered that $C$ is MC rationalized by just the first two priority orderings $(\succeq_1,\succeq_2)$.
That is, the MC rationalization constructed in the proof of Theorem \ref{choice:thm:MCrational} is not necessarily a \emph{minimum size} MC rationalization.
In fact, it is the maximum size MC rationalization in the sense that it includes every linear ordering that can appear in any MC rationalization.

\section{Combinatorial choice from priorities and capacities}
\label{choice:sec:choiceRules}

Choice rules that make use of priorities to ration a scarce discrete resource, such as admissions into a school, have been attractive to market designers.
In this section, we study a few different classes of rules that make use of priorities.

Fix a finite combinatorial choice model $(\X, \mbfC)$ with a complete domain of option sets.
A \emph{priority ordering} $\succeq$ is a complete, transitive, and anti-symmetric binary relation over $\X$, where $a \succ b$ denotes that $a$ is higher priority than $b$.
A \emph{capacity} $q$ is a non-negative integer.

The \emph{priority maximization} choice rule defines for each pair $(q, \succeq)$ a choice function $C^{q, \succeq} \in \mbfC$ as follows:
for each $\bS \in 2^\X$, if $\card{\bS} \leq q$, then $C^{q, \succeq}(\bS) = \bS$, otherwise, $C^{q, \succeq} \in \mbfC = \set{s_1, \dots, s_q} \subset \bS$, where $m < q$ implies $s_m \succ s_q$ and $s \in \bS \sminus \set{s_1, \dots, s_q}$ implies $s_q \succ s$.
Let $\mbfC^{prio}$ be the set of all choice functions defined by the priority maximization choice rule from $(q, \succeq)$ pairs.

A choice function $C \in \mbfC$ is \df{capacity-filling for capacity $q$} if $\card{C(\bS)} = \min\set{\card{\bS},q}$ for every $\bS \in 2^\X$.
It is capacity-filling if there exists a capacity $q$ for which it is capacity-filling.

A choice function $C \in \mbfC$ \df{respects priorities $\succeq$} if for every $\bS \in 2^\X$ and every $a, b \in \bS$, if $a \in C(\bS)$ and $b \nin C(\bS)$, then $a \succ b$.
Respecting priority simply means every chosen element has higher priority than every element available but not chosen.

It is not hard to see that each $C^{q, \succeq}$ is capacity-filling for capacity $q$ and respects priority $\succeq$.
Left for Exercise \ref{ex:fillRespecting} is the straightforward proof that for any $(q, \succeq)$, if $C \in \mbfC$ is capacity-filling for capacity $q$ and respects priorities $\succeq$, then $C = C^{q, \succeq}$.

\begin{theorem}\label{thm:fillRespecting}
A choice function is capacity-filling for capacity $q$ and respects priorities $\succeq$ if and only if it is defined by the priority maximization choice rule from $(q, \succeq)$.
\end{theorem}

A linear order $R$ on $2^\X$ is a \emph{responsive preference over bundles} if for every $\bS \in 2^\X$,
\begin{inparaenum}[(\itshape 1\upshape)]
\item for every $a \in \X$, $\bS \cup \set{a} \mr R \bS$,
\item for every $a, b \in \X \sminus \bS$, $\bS \cup \set{a} \mr R \bS \cup \set{b}$ if and only if $a \mr R b$.
\end{inparaenum}

A choice function is \emph{$q$-responsively rationalized} if there exists a responsive preference $R$ over bundles $2^\X$ and a capacity limit $q \in \Z_+$ such that for every $\bS \in 2^\X$, $C(\bS) = \Top(\beta(\bS, q), R)$, where $\beta(\bS, q) = \set{\bT \subset \bS : \card{\bT} \leq q}$.

\begin{theorem}
A choice function is $q$-responsively rationalized if and only if it is capacity filling for capacity $q$ and respects priorities for some priority ordering $\succeq$.
\end{theorem}

Given choice function $C \in \mbfC$, the \emph{revealed strict priority relation} of $C$, denoted $\succ^*$, is defined by $a \succ^* b$ if and only if there exists $\bS \in 2^\X$ such that $a, b \in \bS$, $a \in C(\bS)$, and $b \nin C(\bS)$.
\medskip
We say $C$ satisfies \idf{Weak axiom of revealed strict priority (WARSPrio)}{Weak axiom of revealed strict priority}\index{WARSPrio} if 
for every $a, b \in \X$,
if $a$ is revealed strictly prioritized by choice function $C$ to $b$, then $b$ is not revealed strictly prioritized by $C$ to $a$.
WARSPrio is equivalent to the requirement of asymmetry of revealed strict priority.

\begin{theorem}
	\label{responsive_charac}
	A choice function is $q$-responsively rationalized if and only if it satisfies WARSPrio and is capacity-filling.
\end{theorem}

\begin{proof}
The necessity of capacity-filling (and hence size monotonicity) and WARSPrio is left to the reader; see Exercise \ref{ex_responsive}.
For the sufficiency part, suppose that $C$ satisfies size monotonicity and WARSPrio.

We show that $\succ^*$, the revealed strict priority of $X$, is transitive.
Take $a_1, a_2, a_3 \in \X$ such that $a_1 \succ^* a_2 \succ^* a_3$.
These are necessarily distinct elements.
Let $\bS \subset \X \sminus \set{a_1, a_2}$ such that $a_1 \in C(\bS \cup \set{a_1, a_2})$ and $a_2 \in \bS \sminus C(\bS  \cup \set{a_1, a_2})$, which is well-defined given $a_1 \succ^* a_2$.
Let $\bT = \bS \sminus C(\bS)$.
So we have $a_1 \in C(\bS)$, $a_2 \in \bT$, and for every $b \in C(\bS)$ and $b' \in \bT$, we have $b \succ^* b'$.
Also, WARSPrio implies $a_3 \nin C(\bS)$, since $a_2 \in \bT$ and $a_2 \succ^* a_3$.
If $a_3 \in \bS$, then we immediately have $a_1 \succ^* a_3$.

So suppose $a_3 \nin \bS$.
If there exists $b \in \bS$ such that $b \in C(\bS)$ and $b \nin C(\bS \cup \set{a_3})$ then WARSPrio implies that for every $b' \in \bT$, $b' \nin C(\bS \cup \set{a_3})$.
This in turn implies $a_3 \nin C(\bS \cup \set{a_3})$, since we are given $a_2 \succ^* a_3$.
But then $C(\bS \cup \set{a_3}) \subset \bS \sminus \bT = C(\bS)$.
Since size monotonicity (and so capacity-filling) implies $\card{C(\bS)} \leq \card{C(\bS \cup \set{a_3})}$, $C(\bS \cup \set{a_3}) = C(\bS)$.
This yields $a_1 \succ^* a_3$.

By the Szpilrajn extension theorem, there exists a linear ordering $\succeq$ such that for each $a,b \in \X$ with $a\neq b$, $a\succ^* b$ implies $a\succ b$.
Let $q = \max\set{ \card{C(\bS)} : \bS \in 2^\X}$.
Necessarily, capacity-filling implies capacity-filling for capacity $q$ as defined.
We now show that for each $\bS\in 2^\X$, $C(\bS)$ is obtained by choosing the highest priority elements according to $\succeq$ until the capacity $q$ is reached or no element is left.
If $|S|\leq q$, by capacity-filling, $C(\bS)=\bS$ and the claim trivially holds.
Suppose that $|S|>q$.
Suppose the claim does not hold. By capacity-filling, this means that there exist $a,b\in \bS$ such that $a\in C(\bS)$, $b\notin C(\bS)$, and $a\succeq b$.
The facts that  $a\in C(\bS)$ and $b\notin C(\bS)$ imply $a\succ^*b$, which contradicts $a\succeq b$ by the construction of $\succeq$.
\end{proof}

WARSPrio does not necessarily imply substitutability when capacity filling is not true.
Exercise \ref{ex1} asks for an example of this.
When capacity filling holds true, substitutability is a weaker property than WARSPrio.

\begin{remark}
Given the discussion in section \ref{choice:sec:combModel}, a priority ordering over $\X$ is not a preference over alternatives, which would have to be a relation over $2^\X$.
In particular, WARP speaks to preferences and WARSPrior to priorities, and their relationship to each other for the combinatorial setting is not immediate.
However, if $\card{C(\bS)}=1$ for all $\bS\in 2^\X$, then we might identify elements of $\X$ with alternatives and recognize they have equivalent implications for such $C$.
Note that Theorem \ref{responsive_charac} connects the question of rationalizability by preferences over alternatives to an axiom (WARSPrio) defined in the language of combinatorial choice.
Since in the present setting WARP is equivalent to IRE by Theorem \ref{thm:WARPequivIRE}, WARSPrio is a stronger requirement than WARP, for capacity filling $C$.
\end{remark}

A \emph{priority sequence} is a sequence of priority orderings on $\X$.
In Section \ref{sec:MCchoice} we saw that priority sequences define path independent choice functions under the MC choice rule, where each priority in the sequence identifies its maximizer in the given option set and the chosen set is just the collection of these maximizers.
An interpretation of this rule is that there is no rivalry between the priorities in the sequence, which takes mathematical expression in the fact that the choice defined by the rule is invariant to permutations in the sequencing of a given priority sequence.

A natural cousin to the MC choice rule is the \emph{sequenced priority maximization with rivalry}.
It defines, for each pair $(q, (\succeq_m)_{m=1}^{m=q})$ of a capacity and a sequence of $q$ priority orderings, a choice function $C^{q, (\succeq_m)_1^q} \in \mbfC$ as follows:
if $\card{S} \leq q$, then $C^{q, (\succeq_m)_1^q}(\bS) = \bS$, otherwise
$
C^{q, (\succeq_m)_1^q} (\bS) = \set{s_1, \dots, s_q}
$
where $s_1 = \Top(\bS, \succeq_1)$ and for each $m \in \set{2, \dots, q}$,
$$
s_m = \Top(\bS \sminus \set{s_1, \dots, s_{m-1}}, \succeq_m).
$$
Let $\mbfC^{seq-prio-riv}$ be the set of all choice functions defined through the sequenced priority maximization with rivalry choice rule.

\begin{theorem} \label{choice:thm:seq-prio-riv}
Every $C \in \mbfC^{seq-prio-riv}$ is capacity-filling and satisfies substitutability.
\end{theorem}
\begin{proof}
Capacity-filling directly follows from the definition of $\mbfC^{seq-prio-riv}$.
To see substitutability, take any $\bS, T\in 2^\X$ and $a\in E$ such that $a\in T\subseteq \bS$.
Suppose that $a\in C(\bS)$. If $|T|<q$, then trivially $a\in C(\bT)$.
Suppose that $|T|\geq q$.
Then, there exists $k \in \{1,\ldots ,q\}$ such that $a$ is the maximizer of $\succeq_k$ when the maximizers of the priority orderings are computed sequentially at $\bS$ as in the definition of the sequenced priority maximization with rivalry choice rule.
Observe that  when the maximizers of the priority orderings are computed sequentially at $\bT\subseteq \bS$, $a$ is either the maximizer of $\succeq_k$ or the maximizer of $\succeq_{k'}$ for some $k'<k$.
Hence, $a\in C(\bT)$.
\end{proof}

\medskip

In some applications, the set of elements $\X$ comes with the structure $(L, \lambda)$, where $\lambda$ is a map from $\X$ \emph{onto} a set $L$.
This structure induces the partition $\set{\X_l}_{l \in L}$, where $\X_l = \set{a \in \X: \lambda(a) = l}$.

Consider again the problem with a capacity $q$ and a single priority ordering $\succeq$ on $\X$.
A \emph{reserves profile} for partition structure $(L, \lambda)$ is an indexed set $\set{r_l: l \in L}$ of non-negative integers that respects capacity, that is $\sum_{l \in L} r_l \leq q$.
The interpretation we pursue here is that $r_l$ is a reserve of capacity for the elements in $\X_l$.

The \emph{reserves-based priority maximization} choice rule defines, for each triple $(q, (r_t)_{t\in T},\succeq)$ of a capacity, a reserves profile, and a priority ordering, a choice function $C^{q, (r_t)_{t\in T},\succeq}$ such that the chosen set from options $\bS$ is determined by the following two-stage procedure:
\begin{enumerate}
\item[\emph{First Stage:}] For each $l \in L$, the $r_l$-highest priority elements that are in both $\X_l$ and $\bS$ are chosen, with all of them chosen if they number no more than $r_l$.
\item[\emph{Second Stage:}] The residual capacity $q'$ is the original $q$ less the number of elements chosen in the first stage.  From amongst the elements of $\bS$ not chosen in the first stage, the $q'$-highest priority ones are chosen (or all of them if they number no more than $q'$).
\end{enumerate}

Let $\mbfC^{res}$ be the set of all choice functions definable by the reserves-based priority maximization choice rule.

\begin{theorem}\label{thm:resIsSubs}
Each $C \in \mbfC^{res}$ is capacity-filling and satisfies substitutability.
\end{theorem}

\begin{proof}
	Capacity-filling follows straightforwardly from the definition of $\mbfC^{res}$.
	To prove substitutability, take any $\bS, \bT\in 2^\X$ and $a\in \X$ such that $a\in \bT\subseteq \bS$ and suppose that $a\in C(\bS)$.
	Let $l = \lambda(a)$.
	If $\card{\bT} < q$, then trivially $a \in C(\bT)$, so suppose that $\card{T} \geq q$.
	
	First suppose that $a$ is chosen in the first stage of the procedure from $\bS$.
	Then $a$ is one of the top $r_l$ elements from $\X_l \cap \bS$ with respect to $\succeq$.
	Since $\bT \subseteq \bS$, $a$ is one of the top $r_l$ elements also from $\X_l \cap T$ with respect to $\succeq$ and chosen in the first stage of the procedure from $\bT$.
	Hence, $a \in C(\bT)$.
	
	What we have just proved is that the choice function defined by the choices from only the first stage of the procedure satisfies substitutability.
	Since at the second stage, the elements under consideration for choice are exactly those present at the start of the first stage but not chosen, substitutability of the first stage implies that no element in $\bT$ that is chosen in the first stage from the option set $\bS$ will be present at the start of the second stage when the procedure applied to $\bT$.
	That is, the set of elements that could be chosen at the start of the second stage is a monotonic function of the option set at the start of the first stage.
	Then $a \in C(\bT)$ for the same reason as Theorem \ref{choice:thm:seq-prio-riv}, which states that choice functions defined by the priority maximization rule satisfy substitutability.
\end{proof}

Let us illustrate the two classes of choice functions above with an example. 

\begin{example}
	Let $\X=\{1,2,3,4,5\}$ be a set of students. 
	Let $L=\{l,m,h\}$ denote a partition of students into low, medium, and high socioeconomic status, such that $\lambda(5) = l$, $\lambda(1) =\lambda(4)=m$, and $\lambda(2)=\lambda(3) = h$.
	Let $\succeq$ be defined as $1 \succ 2 \succ 3 \succ 4 \succ 5$.
	Let $q=3$.
	
	Let $C_1$ be the choice function in $\mbfC^{seq-prio-riv}$ defined by the sequenced priority maximization with rivalry rule using $(\succeq,\succeq^l,\succeq^m)$, where $\succeq^l$ and $\succeq^m$ are obtained from $\succeq$ by moving the low and medium socioeconomic status students to the top of the priority ordering, respectively.
	That is, $3\succ^l 5 \succ^l 1 \succ^l 2 \succ^l 4$ and $4\succ^m 1 \succ^m 2 \succ^m 3 \succ^m 5$.
	So $C_1(\X)$ is obtained by first choosing the highest priority student, who is student $1$, then from the remaining set, choosing the highest priority $l$ student, who is student $5$, and finally, from the remaining set, choosing the highest priority $l$ student, who is student $4$.
	That is, $C_1(\X)=\{1,4,5\}$.
	Note that the way $C_1$ operates features preference for diversity along with preference for respecting $\succeq$.
	
	Let $C_2$ be the choice function in $\mbfC^{res}$ defined by the reserves-based priority maximization rule $r_l = r_m=1$ and $r_h=0$.
	Note that, for example, $C_2(\X)$ is obtained by first choosing one highest priority $l$ student, who is $5$, and one highest priority $m$ student, who is $1$, and then from the remaining set, choosing the highest priority student, who is $2$. 
	That is, $C_2(\X)=\{1,2,5\}$.
	Note that also the way $C_2$ operates features preference for diversity along with preference for respecting $\succeq$.
\end{example}

This example illustrates how choice rules incorporating diversity considerations into choice behavior can differ in the choices made even in a simple case.
In the context of school choice, Chapter ``School Choice" studies this kind of issue.

\medskip

We close this section by addressing the pattern of results in the choice function design approach we have explored in this section.
From examining the proof of Theorem \ref{thm:resIsSubs}, it can be seen that a more general result could be shown by adapting the arguments made.

The \emph{two-stage selection with rivalry choice rule} $\mc H$ is defined for each pair of choice functions $(C_1, C_2)$ by the following choice procedure applied at each option set $\bS$:
\begin{enumerate}
\item[\emph{First Stage:}] Choose those elements in $\bS$ that $C_1$ would choose.
\item[\emph{Second Stage:}] From amongst the elements of $\bS$ not chosen in the first stage, choose those elements that $C_2$ would choose.
\end{enumerate}
For each pair of choice functions $(C_1, C_2)$, let $\mc H[C_1 \rsu C_2]$ be the choice function defined by this rule.

\begin{theorem}\label{thm:rsuSubsFill}
If $C_1$ and $C_2$ are capacity-filling choice functions that satisfy substitutability, then $\mc H[C_1 \rsu C_2]$ is capacity filling and satisfies substitutability.
\end{theorem}

The proof is left as Exercise \ref{ex:rsuSubsFill}.

\section{Choice and Deferred Acceptance}
\label{choice:sec:choiceAndDA}

In one-to-many matching applications where agents are matched with objects, combinatorial choice rules associated with objects to determine matchings are intimately tied with \begin{inparaenum}[(i)]\item the relationships between different stability notions, \item the relationships between different formulations of the deferred acceptance algorithm, and \item the stability properties of the algorithm's outcomes.\end{inparaenum}
We illustrate this in a simple one-to-many matching model.
Many of these results hold for general models of matching with contracts, covered in Chapter ``Generalized Matching".

There is a set of agents $A$ and a set of objects $O$.
Each agent $i\in A$ has a preference relation $R_i$ over $O\cup \{\emptyset\}$ that is complete, transitive, and anti-symmetric, where $\emptyset$ denotes being unmatched (or the outside option).
We also write $a \mr P_i b$ if $a \mr R_i b$ and $a\neq b$.
An object $a\in O$ is \df{acceptable} for $i$ if $a \mr R_i \emptyset$.
Each object $a\in O$ has a choice rule $C_a: 2^A\rightarrow 2^A$ which associates with each set of applicants $\bS \subseteq A$, a nonempty set of chosen agents $C_a(\bS)\subseteq \bS$ with the interpretation that each chosen agent receives a copy of the object.

A matching $\mu$ is an assignment of objects to agents such that each agent receives at most one object.
When agent $i\in A$ is matched with object $a\in O$, we write $\mu(i)=a$ and $i\in \mu(a)$.
An agent $i\in A$ might also be unmatched at $\mu$, which we denote by $\mu(i)=\emptyset$. 

\subsection{Stability}

A central property for matchings is \df{stability}, which roughly requires that the matching is robust to agents possibly acting to circumvent the match.
Even in the context of centralized resource allocation, where the objects are not agents in the usual sense, stability has normative foundations, as discussed in the introduction of the chapter.
Below are some alternative formulations of stability.

A matching $\mu$ is \idf{individually stable}{individual stability} if it is acceptable to every agent $i \in A$ and for every object $a \in O$,  $C_a(\mu(a)) = \mu(a)$.
A matching $\mu$ is \df{$\alpha$-stable} if it is individually stable and there does not exist an agent $i\in A$ and an object $a\in O$  such that $a \mr P_i  \mu(i)$ and $i \in C_a(\mu(a) \cup \{i\})$.
A matching $\mu$ is \df{$\beta$-stable} if it is individually stable and there does not exist an agent $i\in A$ and an object $a\in O$  such that $a \mr P_i  \mu(i)$ and $C_a(\mu(a) \cup \set{i}) \neq \mu(a)$.

Note the difference between the above two versions of stability: in the $\alpha$-version, the agent $i$ who approaches $a$ must be chosen by $a$ in order to successful circumvent the matching , while in the $\beta$-version, it is sufficient if the initial matching of $a$ is disrupted.
Under IRE, these two versions are equivalent.

\begin{lemma}
Without any assumptions on objects' choice rules, $\beta$-stability implies $\alpha$-stability, but not vice versa.
If objects' choice rules satisfy IRE, then $\alpha$-stability is equivalent to $\beta$-stability.
\end{lemma}

\begin{proof}
We can simply assume individual stability holds.
On one hand, if $\mu$ is not $\alpha$-stable, 
then there exist $i \in A$ and $a \in O$ such that $a \mr P_i  \mu(i)$ and $i \in C_a(\mu(a) \cup \set{i})$, which implies that $C_a(\mu(a) \cup \set{i}) \neq \mu(a)$.
Hence, $\mu$ is not $\beta$-stable. 
On the other hand, if $\mu$ is not $\beta$-stable,
then there exist $i \in A$ and $a\in O$ such that $a \mr P_i  \mu(i)$ and $C_a(\mu(a) \cup \set{i}) \neq \mu(a)$.
If $C_a$ satisfies IRE, then $i \in C_a(\mu(a) \cup \{i\})$, and so $\mu$ is not $\alpha$-stable.

An example with two agents and one object can be constructed to show that the equivalence fails without IRE, even with substitutability satisfied.
\end{proof}
	
A matching $\mu$ is \df{group stable} if it is individually stable and there does not exist a nonempty set of agents $\bS \subset A$ and an object $a \in O$ such that $a \mr P_i  \mu(i)$ for each $i \in \bS$, and $\bS \subseteq C_a(\mu(a) \cup \bS)$.

\begin{lemma}
If objects' choice rules satisfy substitutability, $\alpha$-stability is equivalent to group stability.
\end{lemma}
\begin{proof}
We can simply assume individual stability holds.
If $\mu$ is not group stable,
then there exist a nonempty set of agents $\bS \subset A$ and an object $a$ such that $a \mr P_i  \mu(i)$ for each $i\in \bS$, and $\bS \subseteq C_a(\mu(a) \cup \bS)$.
Take any $i\in S$.
By substitutability, $i \in C_a(\mu(a) \cup \set{i})$, so $\mu$ is not $\alpha$-stable.
The reverse direction is immediate.
\end{proof}

Thus, substitutability and IRE give us a single natural theory of stability.
\begin{corollary}
If objects' choice rules satisfy substitutability and IRE, then $\alpha$-stability, $\beta$-stability, and group stability are equivalent.
\end{corollary}

\subsection{Deferred Acceptance}
\label{sec:DefAcc}

Below are two versions of the deferred acceptance algorithm.
There is only one difference between the two algorithms: the choice-keeping deferred acceptance (CK-DA) algorithm only keeps the proposals of the set of chosen applicants from the previous step, while the applicant-keeping version (AK-DA) keeps all proposals received.
The AK-DA algorithm is often called the cumulative offer algorithm (with simultaneous proposals).

\begin{algorithm}[ht]
	\DontPrintSemicolon
	\SetAlgoLined
	\SetKwData{Active}{Active}
	\SetKwFunction{Available}{AvailableTo}
	\SetKwData{Top}{Top}
	\SetKwData{Proposers}{Proposers}
	\SetKwFunction{Pool}{KeptProposals}
	\SetKwFunction{Held}{Held}
	\SetKwFunction{ChoiceFunction}{ChoiceFunction$_o$}
	\SetKwData{Rejected}{Rejected}
	\SetKwData{Chosen}{Chosen}
	\footnotesize
	initialize: every agent $i$ is \Active and every object is \Available{i}, \;
	\While{some agent is \Active}{
		Each \Active agent $i$ proposes to \Top object \Available{i} \;
		\For{each object $o$}{
			\Chosen = \ChoiceFunction{\Proposers $\cup$ \Pool{o}} \;
			Agents not \Chosen are \Rejected if they are \Proposers or \Held{o} \;		
			Each \Rejected agent $i$ no longer has $o$ $\Available{i}$ \;
			Each \Rejected agent $i$ with one or more objects $\Available{i}$ is set as \Active; every other agent is no longer \Active \;
			\Held{o} $\leftarrow$ \Chosen \;
			\If {\normalsize Choice-Keeping DA}{
		\nl		\Pool{o} $\leftarrow$ \Chosen \;
			}
			\If {\normalsize Applicant-Keeping DA}{
		\nl		\Pool{o} $\leftarrow$ \Proposers $\cup$ \Pool{o} \;
			}
		}
	}
	return: each object is assigned to each agent it holds\;
\caption{Deferred acceptance algorithms, CK-DA and AK-DA, differing only at line numbers 1 and 2}
\end{algorithm}

\begin{lemma}
Both CK-DA and AK-DA algorithms stop in finitely many steps without any assumptions on objects' choice rules.
\end{lemma}
This lemma holds since both algorithms terminate when there are no active agents and in each iteration of the proposal phase, either the set of active agents shrinks or an agent's proposal is rejected, shrinking his set of available objects.  
With finite sets of agents and objects, the algorithms terminate.

Do CK-DA and AK-DA algorithms necessarily produce a feasible outcome, that is, a matching?
\begin{lemma}	\label{lemma_feasible_matching}
The CK-DA algorithm produces a matching (a feasible outcome) without any assumptions on objects' choice rules.
The AK-DA algorithm produces a matching if objects' choice rules are substitutable; otherwise, the AK-DA algorithm might not produce a matching.
\end{lemma}

We omit the relatively simple proof.	
However, it is worth understanding the role of substitutability in ensuring a feasible outcome for the AK-DA algorithm.
Without substitutability, it is possible that a proposal by agent $i$ is rejected by an object $a$ but at a later step, because of other proposals to $a$ following $i$'s rejection, $i$ is chosen from the necessarily expanded set of proposals.
In the interim, $i$ may have proposed to another object and be held by it.
If this situation continues to hold at algorithm termination, the outcome is not a matching because it is infeasible for $i$.
With substitutability, the choice functions have the monotone rejection property (see Exercise \ref{choice:ex:equivSubs}) on expanding sets, thereby ensuring a feasible outcome as the set of proposals to each object expands in running of the AK-DA algorithm.

What conditions on objects' choice rules imply that the CK-DA and AK-DA algorithms outcome-equivalent?
Substitutability is not sufficient for this equivalence; see Exercise \ref{exIRE}.
For this equivalence, we need IRE in addition to substitutability.

\begin{lemma}
If objects' choice rules satisfy IRE and substitutability, then the CK-DA and AK-DA algorithms are outcome-equivalent; in fact, after each round of agents proposing and objects choosing, the tentatively accepted (i.e. held) proposals coincide.

If objects' choice rules fail substitutability or IRE, then the outcomes of the CK-DA and AK-DA algorithms might differ.
\end{lemma}

We next investigate the stability properties of the two deferred acceptance algorithms.

\begin{lemma}
If objects' choice rules satisfy IRE and substitutability, then the outcome of the CK-DA algorithm is a stable matching.
If objects' choice rules satisfy IRE, then the outcome of the AK-DA algorithm is a stable matching.
\end{lemma}    

If objects' choice rules fail substitutability or IRE, then the outcome of the CK-DA might be an unstable matching; see Exercise \ref{exStable}. Also, if objects' choice rules fail IRE, then the outcome of the AK-DA might be an unstable matching. The following table summarizes conditions required on objects' choice rules to guarantee feasibility and stability of CK-DA and AK-DA algorithms.

\begin{figure}[ht]
	\begin{center}
		\begin{tabular}{| l | c | c | c |}
			\cline{1-4}
			& \emph{Feasible}                   & \emph{Stable}                      &  $=$ \emph{AK-DA}\\
			\cline{1-4}
			\emph{CK-DA}            & None                       & Path Independence           &  Path Independence \\
			\cline{1-4}
			\emph{AK-DA}            & Substitutability           & IRE                         &  None \\
			\cline{1-4}
		\end{tabular}
	\end{center}
	\caption{Choice condition required for various properties of deferred acceptance outcomes}
\end{figure}

\section{Notes}
\label{choice:sec:notes}

\cite{Moulin:1985scw} is a great survey with some of the results we have presented here, and many worthwhile ones we have not.
\cite{ChamEch:2016book} is a lucid treatment of revealed preference theory with a focus on partial observability of choice behavior and questions of falsifiability and testability.

\cite{Echenique:2007geb} introduces the term combinatorial choice and counts the number of substitutable choice functions.
Theorem \ref{thm:WARPequivIRE} and other behavioral implications for combinatorial choice of some classic requirements such as WARP, Sen's $\alpha$, and Plott's path independence \citep{Plott:1973ecta} are studied in \cite{Alva:2018jet}, by way of the behavioral isomorphism between combinatorial and pure choice models.
Theorem \ref{thm:ratCombDemand} is due to \cite{ChamEch:2018mor}.

The characterization of path independent choice functions by lattices of sets appears in various different forms in the literature since at least \cite{JohnDean:2001mss}.
The lattice structure we describe in Section \ref{choice:sec:latticeMaxOptionSets}, and Theorem \ref{choice:thm:closureLattice}, is due to \cite{Koshevoy:1999mss}.
MC rationalizability and Theorem \ref{choice:thm:MCrational} are due to \cite{AizMal:1981AutCont}.
\cite{doganmor} and \cite{kopylov} investigate the minimum size of an MC rationalization.

Theorem \ref{responsive_charac} on responsive rationalizations of combinatorial choice functions is from \cite{chambers2017simple}.
The sequenced priority maximization with rivalry choice rules were first considered by \cite{KomSon:2016te} in a matching with contracts model (see Chapter ``Matching with Transfers" for the definition of this model).
Reserve-based priority maximization choice rules are axiomatically characterized in \cite{EchYen:2015aer}.

There are two topics omitted in the chapter that we encourage the serious student to explore.
First, the space of path independent choice functions on a given ground set has a remarkable lattice structure of its own \citep{DaniKosh:2005mss} and many choice rules defined in the choice function design literature, including the ones examined in Section \ref{choice:sec:choiceRules}, can be understood as operators on the space of choice functions.
Second, there is a body of work falling under discrete convexity theory that has found increasing success in the study of discrete goods settings with money.
One salient result pertains to combinatorial demand: the combinatorial demand function derived from a valuation function satisfies substitutability (as defined for the matching with salaries model of Chapter ``Matching with Transfers") if and only if the valuation function is M$^\natural$-concave, which is an adaptation of concavity for discrete spaces \citep{FujiYang:2003mor}.
\cite{Murota:2016jmid} offers a tremendous survey of discrete convex analysis with extensive references to applications in economics.

\section{Exercises}
\label{choice:sec:exercises}

\begin{exercise}
\label{ex:WARPimplies}
Complete the proof of Theorem \ref{thm:WARPimplies}, by showing that WARP implies transitive rationalizability when the domain $\mc B$ is additive.
\end{exercise}

\begin{exercise}
\label{exIdem}
A function $f$ mapping a space to itself is \emph{idempotent} if $f$ composed with $f$ is equal to $f$ (i.e. $f \circ f = f$). Show that a combinatorial choice function $C$ is idempotent if it satisfies substitutability or IRE.
\end{exercise}

\begin{exercise}
\label{choice:ex:equivSubs}
For a combinatorial choice function $C$ on $\X$, show that each of the following properties is equivalent to substitutability:
	\begin{enumerate}[a.]
		\item \emph{Subadditivity:} $\forall \bS, \bT$, $C(\bS \cup \bT) \subset C(\bS) \cup C(\bT)$
		\item \emph{Monotone rejection:} $\forall \bS, \bT$, if $\bS \subset \bT$, then $\bS \sminus C(\bS) \subset \bT \sminus C(\bT)$
		\item \emph{Antitone non-rejection:} $\forall \bS, \bT$, if $\bS \subset \bT$, then $C(\bS) \cup (\X \sminus \bS) \supseteq C(\bT) \cup (\X \sminus \bT)$
	\end{enumerate}
\end{exercise}

\begin{exercise}
\label{exPIvariants}
For a contraction mapping $f$ on the powerset $\mc S$ of a set $S$, i.e. a mapping $f$ satisfying $\forall S_1 \in \mc S, f(S_1) \subset S_1$, show that each of the following statements is equivalent to path independence:
	\begin{enumerate}[a.]
		\item $f$ satisfies the equation $f(S_1 \cup S_2) = f(f(S_1) \cup S_2)$ 
		\item $f$ is idempotent and \emph{additive-in-the-image}: $f(f(S_1 \cup S_2)) = f(f(S_1) \cup f(S_2))$
	\end{enumerate}
\end{exercise}

\begin{exercise}
\label{choice:ex:MC}
Show that MC rationalizability implies path independence.
\end{exercise}

\begin{exercise}
\label{ex:fillRespecting}
Complete the proof of Theorem \ref{thm:fillRespecting}.
\end{exercise}

\begin{exercise}
\label{ex_responsive}
Complete the necessity part of the proof for Theorem \ref{responsive_charac}, that if a choice function does not satisfy capacity-filling or WARSPrio, then it is not capacity-constrained responsive.
\end{exercise}

\begin{exercise}
\label{ex1}
Construct an example of a combinatorial choice function that satisfies the WARSPrio but violates substitutability.
\end{exercise}

\begin{exercise}
\label{ex:rsuSubsFill}
Adapt the proof of Theorem \ref{thm:resIsSubs} to prove Theorem \ref{thm:rsuSubsFill}.
\end{exercise}

\begin{exercise}\label{exIRE}  \label{exStable}
For each of the following, find an example of a one-to-many matching problem with the provided information:
	\begin{enumerate}[a.]
	\item the CK-DA and AK-DA algorithms produce different outcomes and objects' choice functions satisfy substitutability.
	\item the CK-DA algorithm produces an unstable outcome and objects' choice functions satisfy IRE
	\end{enumerate}
\end{exercise}

\bibliographystyle{cambridgeauthordate}
\bibliography{references}

\begin{thebibliography}{17}
\expandafter\ifx\csname natexlab\endcsname\relax\def\natexlab#1{#1}\fi
\expandafter\ifx\csname selectlanguage\endcsname\relax
  \def\selectlanguage#1{\relax}\fi

\bibitem[\protect\citename{Aizerman and Malishevski, }1981]{AizMal:1981AutCont}
Aizerman, Mark~A., and Malishevski, Andrey~V. 1981.
\newblock General Theory of Best Variants Choice: Some Aspects.
\newblock {\em IEEE Transactions on Automatic Control}, {\bf 26}(5),
  1030--1040.

\bibitem[\protect\citename{Alva, }2018]{Alva:2018jet}
Alva, Samson. 2018.
\newblock {WARP} and combinatorial choice.
\newblock {\em Journal of Economic Theory}, {\bf 173}(1), 320--333.

\bibitem[\protect\citename{Chambers and Echenique, }2016]{ChamEch:2016book}
Chambers, Christopher~P., and Echenique, Federico. 2016.
\newblock {\em Revealed Preference Theory}.
\newblock Econometric Society Monographs.
\newblock Cambridge University Press.

\bibitem[\protect\citename{Chambers and Echenique, }2018]{ChamEch:2018mor}
Chambers, Christopher~P., and Echenique, Federico. 2018.
\newblock A Characterization of Combinatorial Demand.
\newblock {\em Mathematics of Operations Research}, {\bf 43}(1), 222--227.

\bibitem[\protect\citename{Chambers and Yenmez, }2018]{chambers2017simple}
Chambers, Christopher~P., and Yenmez, M.~Bumin. 2018.
\newblock A simple characterization of responsive choice.
\newblock {\em Games and Economic Behavior}, {\bf 111}, 217--221.

\bibitem[\protect\citename{Danilov and Koshevoy, }2005]{DaniKosh:2005mss}
Danilov, Vladimir, and Koshevoy, Gleb. 2005.
\newblock Mathematics of Plott choice functions.
\newblock {\em Mathematical Social Sciences}, {\bf 49}(3), 245--272.

\bibitem[\protect\citename{Do\u{g}an {et~al.}, }2021]{doganmor}
Do\u{g}an, Battal, Do\u{g}an, Serhat, and Yildiz, Kemal. 2021.
\newblock On Capacity-Filling and Substitutable Choice Rules.
\newblock {\em Mathematics of Operations Research}, {\bf forthcoming}.

\bibitem[\protect\citename{Echenique, }2007]{Echenique:2007geb}
Echenique, Federico. 2007.
\newblock Counting combinatorial choice rules.
\newblock {\em Games and Economic Behavior}, {\bf 58}(2), 231--245.

\bibitem[\protect\citename{Echenique and Yenmez, }2015]{EchYen:2015aer}
Echenique, Federico, and Yenmez, M.~Bumin. 2015.
\newblock How to Control Controlled School Choice.
\newblock {\em American Economic Review}, {\bf 105}(8), 2679--2694.

\bibitem[\protect\citename{Fujishige and Yang, }2003]{FujiYang:2003mor}
Fujishige, Satoru, and Yang, Zaifu. 2003.
\newblock A Note on Kelso and Crawford's Gross Substitutes Condition.
\newblock {\em Mathematics of Operations Research}, {\bf 28}(3), 463--469.

\bibitem[\protect\citename{Johnson and Dean, }2001]{JohnDean:2001mss}
Johnson, Mark~R., and Dean, Richard~A. 2001.
\newblock Locally complete path independent choice functions and their
  lattices.
\newblock {\em Mathematical Social Sciences}, {\bf 42}(1), 53--87.

\bibitem[\protect\citename{Kominers and S{\"o}nmez, }2016]{KomSon:2016te}
Kominers, Scott~Duke, and S{\"o}nmez, Tayfun. 2016.
\newblock Matching with slot-specific priorities: Theory.
\newblock {\em Theoretical Economics}, {\bf 11}(2), 683--710.

\bibitem[\protect\citename{Kopylov, }2021]{kopylov}
Kopylov, Igor. 2021.
\newblock Minimal rationalizations.
\newblock {\em Economic Theory}, {\bf forthcoming}.

\bibitem[\protect\citename{Koshevoy, }1999]{Koshevoy:1999mss}
Koshevoy, Gleb~A. 1999.
\newblock Choice functions and abstract convex geometries.
\newblock {\em Mathematical Social Sciences}, {\bf 38}(1), 35--44.

\bibitem[\protect\citename{Moulin, }1985]{Moulin:1985scw}
Moulin, Herv{\'e}. 1985.
\newblock Choice Functions Over a Finite Set: A Summary.
\newblock {\em Social Choice and Welfare}, {\bf 2}(2), 147--160.

\bibitem[\protect\citename{Murota, }2016]{Murota:2016jmid}
Murota, Kazuo. 2016.
\newblock Discrete convex analysis: A tool for economics and game theory.
\newblock {\em Journal of Mechanism and Institution Design}, {\bf 1}(1),
  151--273.

\bibitem[\protect\citename{Plott, }1973]{Plott:1973ecta}
Plott, Charles~R. 1973.
\newblock Path Independence, Rationality, and Social Choice.
\newblock {\em Econometrica}, {\bf 41}(6), 1075--1091.

\end{thebibliography}

\end{document}